\documentclass[journal,draftclsnofoot,onecolumn,12pt,twoside]{IEEEtran}
\usepackage{amsmath,amsfonts}
\usepackage{algorithmic}
\usepackage{algorithm}
\usepackage{amssymb}
\usepackage{color}

\usepackage{float} 
\usepackage{subfigure}

\usepackage{array}
\usepackage[caption=false,font=normalsize,labelfont=sf,textfont=sf]{subfig}
\usepackage{textcomp}
\usepackage{stfloats}
\usepackage{url}
\usepackage{bm}
\usepackage{verbatim}
\usepackage{multirow} 
\usepackage{ntheorem}
\usepackage{graphicx}
\usepackage{makecell}
\usepackage{cite}
\theorembodyfont{\upshape}
\theoremheaderfont{\it}
\qedsymbol{\square}
\theoremseparator{:}

\newtheorem{lemma}{\indent Lemma}
\newtheorem*{proof}{\indent Proof}

\makeatletter

\newcommand{\Rmnum}[1]{\expandafter\@slowromancap\romannumeral #1@}
\makeatother
\begin{document}

\title{Robust Beamforming Design for RIS-aided Cell-free Systems with CSI Uncertainties and Capacity-limited Backhaul}

\author{Jiacheng~Yao,~\IEEEmembership{Student Member,~IEEE,} 
		Jindan~Xu,~\IEEEmembership{Member,~IEEE,} 
		Wei~Xu,~\IEEEmembership{Senior~Member,~IEEE,}
		Derrick~Wing~Kwan~Ng,~\IEEEmembership{Fellow,~IEEE,}
		Chau~Yuen,~\IEEEmembership{Fellow,~IEEE,} 
		and~Xiaohu~You,~\IEEEmembership{Fellow,~IEEE}% <-this % stops a space
\thanks{Part of this paper was presented in \emph{IEEE ISWCS}, Hangzhou, China, Oct. 2022 \cite{me}.}
\thanks{J. Yao, W. Xu, and X. You are with the National Mobile Communications Research Laboratory (NCRL), Southeast University, Nanjing 210096, China, and also with Purple Mountain Laboratories, Nanjing 211111, China (\{jcyao, wxu, xhyu\}@seu.edu.cn).}
\thanks{J. Xu and C. Yuen are with Engineering Product Development (EPD) Pillar, Singapore University of Technology and Design, Singapore 487372, Singapore (jindan\_xu@sutd.edu.sg, yuenchau@sutd.edu.sg).}
\thanks{Derrick Wing Kwan Ng is with the School of Electrical Engineering and Telecommunications, University of New South Wales, Sydney, NSW 2052, Australia (w.k.ng@unsw.edu.au).}}% <-this % stops a space
%\thanks{Manuscript received April 19, 2021; revised August 16, 2021.}


\maketitle

\begin{abstract}

In this paper, we consider the robust beamforming design in a reconfigurable intelligent surface (RIS)-aided cell-free (CF) system considering the channel state information (CSI) uncertainties of both the direct channels and cascaded channels at the transmitter with capacity-limited backhaul. We jointly optimize the precoding at the access points (APs) and the phase shifts at multiple RISs to maximize the worst-case sum rate of the CF system subject to the constraints of maximum transmit power of APs, unit-modulus phase shifts, limited backhaul capacity, and bounded CSI errors. By applying a series of transformations, the non-smoothness and semi-infinite constraints are tackled in a low-complexity manner that facilitates the design of an alternating optimization (AO)-based iterative algorithm. The proposed algorithm divides the considered problem into two subproblems. For the RIS phase shifts optimization subproblem, we exploit the penalty convex-concave procedure (P-CCP) to obtain a stationary solution and achieve effective initialization. For precoding optimization subproblem, successive convex approximation (SCA) is adopted with a convergence guarantee to a Karush-Kuhn-Tucker (KKT) solution. Numerical results demonstrate the effectiveness of the proposed robust beamforming design, which achieves superior performance with low complexity.
Moreover, the importance of RIS phase shift optimization for robustness and the advantages of  distributed RISs in the CF system are further highlighted.
\end{abstract}


\section{Introduction}
A cell-free (CF) system, that is a user-centric network paradigm, has recently attracted great attention to satisfy unprecedented growing demands for next-generation wireless networks\cite{cf0,zongshu}. Unlike the conventional architecture of cellular networks, a large number of serving antennas, known as access points (APs), are distributed over a wide service area to mitigate the negative influence of shadow fading and to shorten communication distances between transceivers. In particular, all the APs in a CF system connect to a central processing unit (CPU) via backhaul links and these APs serve all users simultaneously with dynamic cooperation \cite{cf1,cf2,cf4}. Besides, by combining the advantages of massive multiple-input multiple-output (MIMO) and network MIMO, a CF system increases the network capacity substantially \cite{cf3}. Moreover, CF systems are free of inter-cell interference, which avoid any potential poor cell-edge user performance since the cell boundaries are eliminated. \par

Despite its significant advantages, there are various new technical challenges in CF systems, such as demanding backhaul capacity, high energy consumption, and huge computational complexity. For CF deployments, in practice, a large amount of data need to be conveyed through backhaul links to ensure efficient cooperation between all APs \cite{cf0}. Limited backhaul capacity is therefore a key bottleneck in practice that restricts the performance of CF systems to a great extent. In \cite{backhaul0,backhaulmodel,backhaul1}, these works have considered the beamforming design for spectral efficiency maximization with  backhaul capacity constraints for CF systems.\par

While achieving better performance, the energy consumption in a CF system is also a serious concern in practice due to the deployment of a large number of APs. Reconfigurable intelligent surface (RIS) has been a promising supplement to help unlock the potential of CF systems with reduced energy and cost \cite{george1,george2}. Specifically, the RIS is able to adaptively manipulate the electromagnetic wave propagation environment and to help decrease the overall system energy consumption by equipping a large number of low-cost passive reflecting elements \cite{ris1,ris3,ris4,ris2,wshi}. Indeed, the deployment of RIS in CF systems has been demonstrated as an effective approach for enhancing the performance of wireless communication in a number of recent works \cite{dai,energycf,dcf,cbf,twotime1,twotime2}. In \cite{dai}, by replacing some APs with low-cost RISs, the authors proposed a joint precoding and RIS phase optimization framework for maximizing the system weighted sum rate, leading to significant network capacity improvement. In \cite{energycf}, the energy efficiency was maximized by the optimization of hybrid beamforming in a RIS-aided CF system. It was shown that the system energy efficiency can be improved substantially through the introduction of RIS. Particularly, the performance gain depends on the number of RISs and the physical size of each RIS. Further, a decentralized beamforming scheme was proposed in \cite{dcf} for a RIS-aided CF system, which asymptotically approaches the performance of a centralized design. In \cite{cbf}, the authors considered the cooperative beamforming (CBF) besign for a RIS-aided CF system, where the hybrid beamforming at BSs and passive beamforming at RISs were jointly optimized to enhance the spectral efficiency. In addition, a two-timescale transmission design was considered in \cite{twotime1,twotime2} for RIS-aided CF systems. The considerable performance gains from the adoption of the CF paradigm and the deployment of RISs were revealed through theoretical and simulation results. It has been widely expected that RIS-aided CF is a key to exhibit higher spectral and energy efficiencies in future wireless communication networks for supporting broad emerging applications.\par

It is worth noting that the beamforming design for RIS-aided CF systems in most of these works, e.g., \cite{dai,energycf,dcf,cbf}, optimistically assumed the availability of perfect channel state information (CSI), which is however impossible to acquire in practice due to limited system resources. In fact, the passive nature of RIS does not facilitate any signal transmission, reception, and processing. Besides, the huge signaling overhead brought by the large number of reflecting elements makes it challenging to estimate channels involving RISs \cite{estimation}. Specifically, without employing any active elements at RIS for advanced signal processing, the cascaded channel of the AP-RIS link and the RIS-user link is usually estimated as an alternative \cite{cascaded,decompose,zhang}. To reduce the huge signaling overhead of training pilots, a compressed channel estimation method exploiting the sparsity in millimeter-wave (mmWave) propagations was proposed in \cite{compress}. In addition, the authors in \cite{fenzu} proposed a novel reflection pattern at the RIS to simplify the beamforming and channel estimation design, where the reflecting elements of RIS were grouped and group-based phase shift pattern was optimized. Apart from these typical passive RIS, hybrid reflecting and sensing RISs (HRISs), which allows the sensing of the impinging signal, is an alternative architecture \cite{george3}. Thanks to these sensing capability of HRIS, channel estimation for RIS-related channels is greatly facilitated \cite{george4}. Despite these fruitful results, only partial CSI is available in practice and the presence of channel estimation error at the transmitter is generally inevitable.\par

Regarding the consideration of partial CSI, there have been a number of studies on robust precoding design especially for multicell networks, e.g., in \cite{mcell1,mcell2,cui}. However, 
they did not consider RIS phase shift optimization, which makes it impossible to reap the performance gain of RIS deployment. Moreover, it has been evidenced  in \cite{mcell1,mcell2,cui} that these worst-case optimization with imperfect CSI is fairly involved even for typical optimization variables and convex constraints, putting aside the additional nonconvex unit-modulus phase shift variables of RISs. As for RIS-aided centralized wireless networks, robust beamforming design with imperfect CSI have been studied in \cite{pan,green,hu,yu}. However, the prior works mostly considered the deployment of a single AP with the assistance of a single RIS, which can hardly guarantee the performance for users distributed over a large area. Unlike the centralized architecture, the performance of the CF system is also greatly limited by the backhaul capacity. Hence, considering the design with the limited backhaul capacity is of practical significance to improve the performance of RIS-aided CF systems. As such, in \cite{backhaulcf}, the authors maximized the energy efficiency for a RIS-aided CF system with the limited backhaul capacity constraints under the assumption of perfect CSI. However, an effective robust beamforming design for RIS-aided CF systems against CSI imperfectness with capacity-limited backhaul is still open.

In this paper, we investigate the robust beamforming design for a RIS-aided CF system with capacity-limited backhaul subject to CSI uncertainties. The bounded CSI error model of both the direct channel and cascaded channels at the transmitter is adopted. We consider the optimization of
the phase shifts at the distributed RISs and the precoding at APs by maximizing the worst-case sum rate of the system, subject to the individual transmit power constraint of AP, the maximum backhaul capacity constraint, and the unit-modulus constraints of phase shifts. The main contributions of this paper are summarized as follows:
\begin{itemize}
\item To tackle the non-smoothness caused by the $l_0$-norm in the backhaul constraints, we exploit the arctangent function to derive an accurate approximation of the $l_0$-norm. As for the semi-infinite constraints brought by the CSI errors, we propose a novel transformation scheme with much lower computational complexity than the traditional S-procedure based scheme\cite{pan,hu}. Specifically, we first derive a closed-form expression to characterize the worst-case value of the desired signal strength in order to simplify the worst-case term in the constraints. On the other hand, as for the interference term, an upper bound is derived  to deal with the infinite number of constraints. Through this efficient transformation, the complexity is successfully reduced by multiple orders-of-magnitude, further promoting the deployment of large-scale RIS-aided CF systems.

\item Due to the coupled AP precoding and RIS phase shifts, we propose an iterative algorithm via alternating optimization (AO) based on the derived transformations to maximize the worst-case sum rate. For the phase shift optimization subproblem involving nonconvex unit-modulus constraints, we use the penalty convex-concave procedure (P-CCP) with a stationary solution, which avoids the challenging requirement of finding a feasible initial point. For the precoding optimization subproblem, we exploit the successive convex approximation (SCA) method to tackle the nonconvexity and obtain a Karush-Kuhn-Tucker (KKT) solution. Moreover, considering possible violations of the backhaul constraints caused by the approximation, we further propose an efficient refinement approach to ensure a feasible solution.
\item We verify the effectiveness of the proposed robust beamforming design for the RIS-aided CF system via numerical results. In particular, compared with the traditional S-procedure based method, the proposed algorithm achieves rapid convergence with marginal performance loss. Meanwhile, the optimization for RIS phase shifts not only provides performance gain, but also improves the robustness against the CSI imperfection. It is found that accurate estimation of the cascaded channel plays a more critical role than that of the direct channel in achieving satisfactory performance. A distributed deployment of RISs is more suitable for CF systems in practice than a large-scale but centralized RIS deployment.
\end{itemize}

The rest of this paper is organized as follows. In Section~\Rmnum{2}, we discuss the system model of the RIS-aided CF system and characterize the CSI error model and backhaul capacity model. Based on these models, we formulate the robust beamforming design problem. In Section \Rmnum{3}, we discuss the transformations of the original problem and propose an iterative algorithm to solve the formulated problem. Finally, numerical results and conclusions are provided in Sections \Rmnum{4} and \Rmnum{5}, respectively.

\textit{Notations:} $\mathbb{C}$ denotes the complex-valued space. $\mathbb{E}\{\cdot\}$ denotes the expectation operation. $(\cdot)^T$, $(\cdot)^*$, and  $(\cdot)^H$ denote the transpose, conjugate, and conjugat transpose operations, respectively. $\Re \{\cdot\}$ and $\Im \{\cdot\}$ denote the real part and the imaginary part of an input complex number, respectively. $\angle(\cdot)$ denotes the phase of a complex number or a complex vector. $\Vert \cdot \Vert_0$ represents the $l_0$-norm. $\vert \cdot \vert$,  $\Vert \cdot \Vert_2$, and $\Vert \cdot \Vert_F$ denote the modulus of a complex number,  Euclidean norm, and Frobenius norm of matrices (or vectors), respectively. Operator ${\rm{diag}}\{\cdot\}$ denotes the diagonal operation. 
%$\bm{I}_K$ and $\bm{0}_K$ stand for a $K\times K$ identity matrix and an all-zero matrix, respectively.

\section{System Model}
\subsection{Signal Model}
In this paper, we consider a RIS-aided cell-free multiple-input single-output (MISO) system, which is shown in Fig. 1. In this system, $K$ single-antenna users are served by $N$ APs, each equipped with $N_t$ antennas, with the assistance of $L$ distributed RISs. Each of the RISs is equipped with $M$ reflecting elements. In a cell-free system, all the APs are connected to a CPU through capacity-limited backhaul links, which is responsible for the calculation of resource allocation and scheduling. The RISs are connected to the CPU or APs and their phase shift configuration are controlled by the CPU or APs. We assume that global CSI is available at the CPU and centralized optimization is conducted by the CPU, which typically presents upper bounded performance of practical implementations of RIS-aided CF systems \cite{cpu}.

\begin{figure}[!t]
\centering
\includegraphics[width=3.8in]{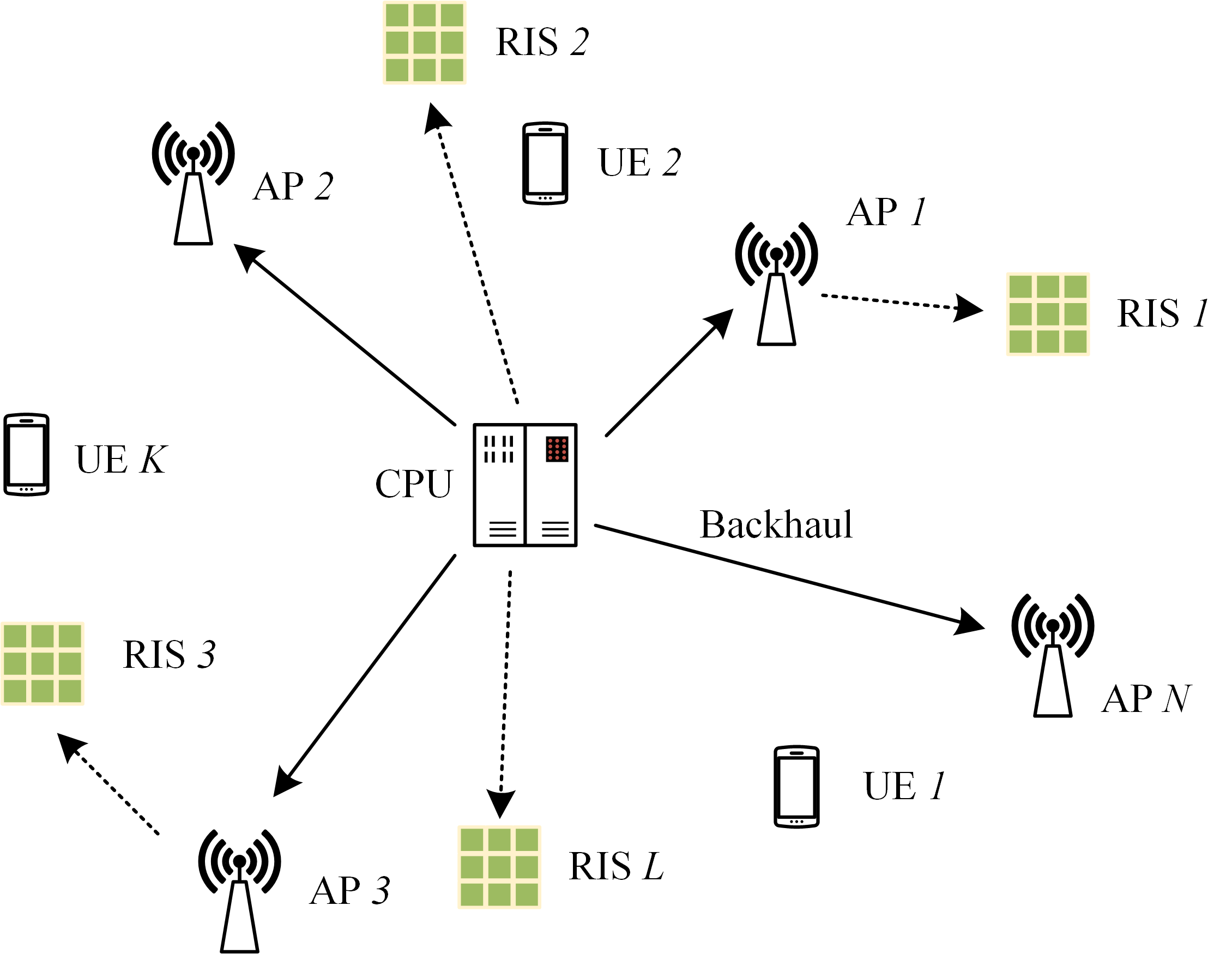}
\caption{A RIS-aided CF system with $N$ APs, $K$ UEs, and $L$ distributed RISs.}
\end{figure}
The channel between the $n$th AP and the $k$th user includes a direct link and $L$ reflecting paths. Let $\bm{h}_{d,n,k} \in \mathbb{C}^{N_t\times1} $, $\bm{h}_{r,l,k}\in \mathbb{C}^{M\times1}$, and $\bm{G}_{n,l}\in \mathbb{C}^{M\times N_t}$, $\forall n \in \{ 1,\cdots, N\}$, $\forall l \in \{ 1,\cdots, L\}$, $\forall k \in \{ 1,\cdots, K\}$, denote the direct channel between the $n$th AP and the $k$th user, the reflecting channel between the $l$th RIS and the $k$th user, and the channel between the $n$th AP and the $l$th RIS, respectively. The phase shift control at the $l$th RIS is denoted by $\bm{\Theta}_l={\rm{diag}}\{ e^{j\theta_{l,1}},\cdots,e^{j\theta_{l,M}}\}$, where ${\rm{e}}^{j\theta_{l,m}}$ represents the phase shift induced by the $m$th antenna element of the $l$th RIS and we ignore the amplitude reflection coefficient induced by RIS for simplicity. Due to severe power loss, we neglect the signals reflected by the RISs for more than once \cite{dai}.  Also,  different propagation delays due to multiple RISs are also neglected as they are generally much shorter than the symbol duration \cite{green}. Hence, the equivalent channel between the $n$th AP and the $k$th user is expressed as
\begin{align}\label{eq1}
\bm{h}_{n,k}^H&=\bm{h}_{d,n,k}^H + \sum_{l=1}^L \bm{h}_{r,l,k}^H \bm{\Theta}_l \bm{G}_{n,l} \triangleq \bm{h}_{d,n,k}^H + \sum_{l=1}^L \bm{v}_l^H \bm{Z}_{n,l,k},
\end{align}
where $\bm{v}_l\triangleq[{\rm{e}}^{j\theta_{l,1}},\cdots,{\rm{e}}^{j\theta_{l,M}}]^H $, and $\bm{Z}_{n,l,k}\triangleq {\rm{diag}}\{\bm{h}_{r,l,k}^H \}  \bm{G}_{n,l}$ represents the cascaded channel between the $n$th BS and the $k$th user through the $l$th RIS. For notational simplicity, we further define $\bm{v}\triangleq [\bm{v}_{1}^H,\cdots,\bm{v}_L^H]^H$ and $\bm{Z}_{n.k}\triangleq [\bm{Z}_{n,1,k}^H,\cdots,\bm{Z}_{n,L,k}^H]^H$. Then, the equivalent channel in (\ref{eq1}) is rewritten as
\begin{equation}
\bm{h}_{n,k}^H =\bm{h}_{d,n,k}^H + \bm{v}^H \bm{Z}_{n,k}.
\end{equation}\par
Then, we write the received signal at the $k$th user as
\begin{align}\label{eq3}
y_k =\sum_{n=1}^N\left(\bm{h}_{d,n,k}^H + \bm{v}^H \bm{Z}_{n,k} \right)\bm{w}_{n,k} s_{k} +\sum_{j\neq k}^K \sum_{n=1}^N  \left(\bm{h}_{d,n,k}^H + \bm{v}^H \bm{Z}_{n,k} \right)\bm{w}_{n,j} s_{j}+ n_k,
\end{align}
where $\bm{w}_{n,k}\in \mathbb{C}^{N_t\times 1}$ represents the beamforming vector at the $n$th BS for the $k$th user, $s_{k}\in \mathbb{C}$ is the symbol transmitted to the $k$th user satisfying $\mathbb{E}\{\vert s_k \vert^2\}=1$ and $\mathbb{E}\{ s_k s_j^*\}=0$, $ \forall k\neq j$, and $n_k$ denotes the additive Gaussian noise with zero mean and variance $\sigma_k^2$, i.e., $n_k \sim \mathcal{CN}(0,\sigma_k^2)$. The received signal consists of both the desired signal and interference from other users.

According to the system model in (\ref{eq3}), the SINR at the $k$th user can be formulated as follows
\begin{equation}
{\rm{SINR}}_k=\frac{ \left \vert \sum_{n=1}^N \left (\bm{h}_{d,n,k}^H +\bm{v}^H \bm{Z}_{n,k} \right)\bm{w}_{n,k} \right \vert^2}{\sum_{j\neq k}^K \left \vert \sum_{n=1}^N \left(\bm{h}_{d,n,k}^H + \bm{v}^H \bm{Z}_{n,k} \right)\bm{w}_{n,j} \right \vert ^2 +\sigma_k^2   },
\end{equation}
and thus the achievable rate of the $k$th user can be evaluated as
\begin{equation}
R_k={\rm{log_2}}(1+{\rm{SINR}}_k).
\end{equation}

\subsection{Channel Uncertainty Model}
In practice, only partial and imperfect CSI is available at the transmitter. In this paper, we adopt a bounded model to characterize the CSI imperfection, which is a general model that isolates the resource allocation design from the specific channel estimation design\cite{hu}. Specifically, the bounded error model is suitable to capture different types of CSI errors, e.g., due to noisy channel estimation, quantization, limited feedbacks, and other possible factors \cite{green,cui}, which is appealing for large-scale distributed systems such as RIS-aided CF systems. Considering that the channel consists of the direct link and cascaded channels, there exist different channel uncertainties. Specifically, the direct link, $\bm{h}_{d,n,k}$, and the cascaded channel, $\bm{Z}_{n,k}$, are, respectively, modeled as
\begin{align}\label{csimodel}
\bm{h}_{d,n,k}&=\hat{\bm{h}}_{d,n,k}+\Delta \bm{h}_{d,n,k},\enspace \Vert \Delta \bm{h}_{d,n,k} \Vert_2 \leq \epsilon_{d,n,k},\nonumber \\
\bm{Z}_{n,k}&=\hat{\bm{Z}}_{n,k}+\Delta \bm{Z}_{n,k},\enspace \Vert \Delta \bm{Z}_{n,k} \Vert_F \leq \epsilon_{c,n,k},
\end{align}
where $\hat{\bm{h}}_{d,n,k}$ and $\hat{\bm{Z}}_{n,k}$ are the estimates of $\bm{h}_{d,n,k}$ and $\bm{Z}_{n,k}$, respectively, and the norms of the unknown CSI errors $\Delta \bm{h}_{d,n,k}$ and $\Delta \bm{Z}_{n,k}$ are limited in the uncertainty regions of constant radii $\epsilon_{d,n,k}$ and $\epsilon_{c,n,k}$, respectively.

\subsection{Backhaul Capacity}
Instead of considering the availability of an infinite-capacity backhaul, we assume that the backhaul link between the CPU and the $n$th AP has the maximal capacity denoted by $C_n$. The data rate conveyed over the backhaul link between the CPU and the $n$th AP is modeled as the sum of the achievable rates of all the users served by the $n$th AP\cite{backhaul1,backhaulmodel}. Moreover, the backhaul capacity should be at least $\xi_n\geq1$ times larger than the data rate transmitted via the $n$th backhaul link maximal capacity  to ensure a feasible transmission\cite{backhaul2}. Hence, we formulate the backhaul capacity constraints under the worst-case transmission as 
\begin{align}
\sum_{k=1}^K \left\Vert \Vert \bm{w}_{n,k}\Vert^2\right\Vert_0\min_{\{\Delta\bm{h}_{d,n,k}\}_{n=1}^N,\, \{\Delta\bm{Z}_{n,k}\}_{n=1}^N}R_k\leq \frac{C_n}{\xi_n},\enspace n=1,2,\cdots,N.
\end{align}
Based on this model, the backhaul capacity constraint is satisfied via the design of beamforming and reducing the number of users being served by each AP. In particular, if the CPU does not forward the $k$th user's data to the $n$th AP, the $n$th AP is not able to serve the $k$th user, i.e., $\Vert \bm{w}_{n,k}\Vert^2=0$, and the $l_0$ norm of $\Vert \bm{w}_{n,k}\Vert^2$ equals to 0, otherwise equals to 1. 

\subsection{Problem Formulation}
In this paper, we aim to maximize the worst-case sum rate of all the users by jointly optimizing both the precoding of APs and the phase shift of the RISs and taking into account the impacts of CSI uncertainties and limited backhaul capacity. The optimization problem is formulated as
\begin{align}\label{problem}
(\mathrm{P}1):\enspace \mathop{{\rm{maximize}}}_{\bm{W},\,\bm{v}} &\quad   \sum_{k=1}^K \min_{ \{\Delta\bm{h}_{d,n,k}\}_{n=1}^N,\, \{\Delta\bm{Z}_{n,k}\}_{n=1}^N}R_k \nonumber \\
{\rm{s.t.}}&\quad {\rm{C}}_1:\enspace \sum_{k=1}^K \Vert \bm{w}_{n,k}\Vert_2^2 \leq P_{n}, \enspace n=1,2,\cdots,N, \nonumber\\
&\quad {\rm{C}}_2:\enspace \vert [\bm{v}]_{m}\vert^2 =1,\enspace m=1,2,\cdots,ML,  \nonumber\\
&\quad {\rm{C}}_3:\enspace \sum_{k=1}^K \left\Vert \Vert \bm{w}_{n,k}\Vert^2\right\Vert_0\min_{\{\Delta\bm{h}_{d,n,k}\}_{n=1}^N,\, \{\Delta\bm{Z}_{n,k}\}_{n=1}^N}R_k\leq \frac{C_n}{\xi_n},\enspace n=1,2,\cdots,N,
\end{align}
where $P_{n}$ is the maximum power constraint of the $n$th AP, $\bm{W}\triangleq \{\bm{w}_{n,k}\}$ is the set of all precoding vectors, and $[\bm{v}]_{m}$ is the $m$th element of $\bm{v}$. The constraint in ${\rm{C}}_2$  represents the unit-modulus constraint of each reflecting RIS element. \par
 
It is worth noting that the problem ($\mathrm{P1}$) is neither convex nor smooth and thus challenging to be solved. Particularly, the precoding matrix, $\bm{W}$, and the phase shift vector, $\bm{v}$, are coupled so that it further complicates their designs. In the following sections, we first provide useful transformations and approximations to simplify the problem at hand that pave the way for developing a computationally efficient algorithm to solve it.

\section{Proposed Robust Beamforming Design}
Due to the intractable form of ($\mathrm{P1}$), we first propose a transformation to recast the problem to a more tractable  and more computationally efficient form. Then, we develop a suboptimal AO-based algorithm to divide the transformed problem into two subproblems, i.e., the phase shift optimization subproblem and the precoding optimization subproblem, which are respectively solved by the P-CCP and SCA technique. In Table \ref{table:2}, we mainly summarize the technical details of the proposed algorithm.

\subsection{Problem Transformation}
To begin with, we deal with the nonconvex objective function. By introducing the slack optimization variables $\{\gamma_k\}_{k=1}^K$ to replace the worst-case SINR terms, we reformulate the original problem ($\mathrm{P}1$) as
\begin{align}
\mathop{{\rm{maximize}}}_{\bm{W},\,\bm{v},\,\bm{\gamma}} &\quad  \sum_{k=1}^K \log_2(1+\gamma_k) \nonumber \\
{\rm{s.t.}}&\quad {\rm{C}}_1,\enspace{\rm{C}}_2,  \nonumber\\
&\quad \mathrm{C}_3:\enspace \sum_{k=1}^K \left\Vert \Vert \bm{w}_{n,k}\Vert_2^2\right\Vert_0 \log_2(1+\gamma_k)\leq \frac{C_n}{\xi_n},\enspace n=1,2,\cdots,N,\nonumber \\
&\quad {\rm{C}}_4:\enspace \min_{  \{\Delta\bm{h}_{d,n,k}\}_{n=1}^N,\, \{\Delta\bm{Z}_{n,k}\}_{n=1}^N}{\rm{SINR}}_k \geq \gamma_k, \enspace k=1,2,\cdots,K.
\end{align}
In addition to the nonconvexity, constraint $\mathrm{C}_3$ is nonsmooth due to the $l_0$-norm and constraint $\mathrm{C}_4$ contains semi-infinite constraints due to the CSI uncertainties. Then, we tackle the non-smoothness and the semi-infinite constraints in the following.

\begin{table}[!t]   
\begin{center}   
\caption{Main Techniques of the Proposed Robust Beamforming Design}  
\label{table:2} 
\begin{tabular}{|p{1.2cm}|p{3cm}|p{4.2cm}|p{3.4cm}|p{2.1cm}|}   
\hline   \textbf{Objective}         & \textbf{Difficulty}   & \textbf{Relaxation and transformation}  & \textbf{Algorithms and features}&\textbf{Complexity}\\
\hline   \multirow{4}{1.2cm}{($\mathrm{P}3$)}& \multirow{5}{3cm}{semi-infinite constraints and  unit-modulus constraint}  & $\mathrm{C}_3 \to \bar{\mathrm{C}}_3 \to \tilde{\mathrm{C}}_3$, &\multirow{5}{3.4cm}{P-CCP, with a stationary solution and effective initialization (Alg. 1)} &  \multirow{5}{*}{$\mathcal{O}(KM^{4.5}L^{4.5})$} \\
& &$\mathrm{C}_4\to\bar{\mathrm{C}}_4  \to \mathrm{C}_{5-7}$&&\\
& &$\mathrm{C}_5\to \bar{\mathrm{C}}_5 \to \tilde{\mathrm{C}}_5$,&&\\
& &${\mathrm{C}}_6 \to \bar{\mathrm{C}}_6$,&&\\
& &$\mathrm{C}_2\to \mathrm{C}_{11-12}\to\bar{\mathrm{C}}_{11-12} $&&\\
\hline   \multirow{3}{1.2cm}{($\mathrm{P4}$)}& \multirow{3}{3cm}{Non-smoothness, and nonconvexity}  & $\mathrm{C}_3\to\bar{\mathrm{C}}_3 \to \mathrm{C}_{8-10}\to \hat{\mathrm{C}}_{8-10}$ &\multirow{3}{3.4cm}{SCA, with a KKT point (Alg. 2)} &  \multirow{3}{*}{$\mathcal{O}(K^{5.5}N^{4.5}N_t^{4})$} \\
& &$\mathrm{C}_5\to \bar{\mathrm{C}}_5 \to \hat{\mathrm{C}}_5$&&\\
& &$\mathrm{C}_6\to \bar{\mathrm{C}}_6$&&\\
\hline \multicolumn{1}{|m{1.2cm}|}{($\mathrm{P}5,6$)}& \multicolumn{1}{m{3cm}|}{Possible infeasible solution} &$\mathrm{C}_3 \to {\mathrm{C}}_{13}$  & \multicolumn{1}{m{3.4cm}|}{Solving ($\mathrm{P5}$) and ($\mathrm{P6}$) via P-CCP and SCA, with a strictly feasible solution}&\multicolumn{1}{m{2.1cm}|}{$\mathcal{O}(KM^{4.5}L^{4.5})$, $\mathcal{O}(K^{5.5}N^{4.5}N_t^{4})$}\\
\hline 

\end{tabular}   
\end{center}   
\end{table}

\subsubsection{$l_0$-norm Approximation}
To further make the problem tractable, we focus on the nonsmooth $l_0$-norm in the backhaul constraint $\mathrm{C}_3$. By means of the arctangent smooth function in \cite{tao}, we can approximate $\left\Vert \Vert \bm{w}_{n,k}\Vert_2^2\right\Vert_0$ as 
\begin{align}\label{appro}
\left\Vert \Vert \bm{w}_{n,k}\Vert_2^2\right\Vert_0\approx \frac{2}{\pi} \mathrm{arctan}\left( \frac{ \Vert \bm{w}_{n,k}\Vert_2^2}{\varpi} \right)\triangleq f_{\varpi}\left(  \Vert \bm{w}_{n,k}\Vert_2^2 \right ),
\end{align}
where $\varpi>0$ is a predetermined parameter controlling the accuracy of the approximation. The smaller $\varpi$, the more accurate the approximation is. It is worth noting that $f_{\varpi}(\cdot)$ is a smooth and concave function for nonnegative input arguments \cite{tao}. Then, we obtain an alternative to constraint $\mathrm{C}_3$ as follows
\begin{align}
\bar{\mathrm{C}}_3:\enspace \sum_{k=1}^Kf_{\varpi}\left(  \Vert \bm{w}_{n,k}\Vert_2^2 \right ) \log_2(1+\gamma_k)\leq \frac{C_n}{\xi_n},\enspace n=1,2,\cdots,N,
\end{align}
which is a smooth but nonconvex. Its nonconvexity will be addressed in the next subsection.

\subsubsection{Infinite Constraints Reformulation}
The worst-case SINRs in constraint $\mathrm{C}_4$ is intractable due to the existence of nonconvexity and infinite constraints. It is easy to first give a lower bound for the worst-case SINR at the $k$th user as
\begin{align}
 \min_{  \{\Delta\bm{h}_{d,n,k}\}_{n=1}^N,\, \{\Delta\bm{Z}_{n,k}\}_{n=1}^N}{\rm{SINR}}_k \geq \frac{ \min_{  \{\Delta\bm{h}_{d,n,k}\}_{n=1}^N,\, \{\Delta\bm{Z}_{n,k}\}_{n=1}^N} \left \vert \sum_{n=1}^N \left (\bm{h}_{d,n,k}^H +\bm{v}^H \bm{Z}_{n,k} \right)\bm{w}_{n,k} \right \vert^2}{\max_{  \{\Delta\bm{h}_{d,n,k}\}_{n=1}^N,\, \{\Delta\bm{Z}_{n,k}\}_{n=1}^N} \sum_{j\neq k} \left \vert \sum_{n=1}^N \left (\bm{h}_{d,n,k}^H +\bm{v}^H \bm{Z}_{n,k} \right)\bm{w}_{n,j} \right \vert^2+\sigma_k^2}.
\end{align}
Due to the nonconvexity of the SINRs, we replace it with this lower bound and thus obtain a performance lower bound of the original problem. Constraint $\mathrm{C}_4$ is replaced by
\begin{align}
\bar{\mathrm{C}}_4:\enspace \frac{ \min_{  \{\Delta\bm{h}_{d,n,k}\}_{n=1}^N,\, \{\Delta\bm{Z}_{n,k}\}_{n=1}^N} \left \vert \sum_{n=1}^N \left (\bm{h}_{d,n,k}^H +\bm{v}^H \bm{Z}_{n,k} \right)\bm{w}_{n,k} \right \vert^2}{\max_{  \{\Delta\bm{h}_{d,n,k}\}_{n=1}^N,\, \{\Delta\bm{Z}_{n,k}\}_{n=1}^N} \sum_{j\neq k} \left \vert \sum_{n=1}^N \left (\bm{h}_{d,n,k}^H +\bm{v}^H \bm{Z}_{n,k} \right)\bm{w}_{n,j} \right \vert^2+\sigma_k^2} \geq \gamma_k, \enspace \forall k.
\end{align}
Then we further split constraint $\bar{\mathrm{C}}_4$ into the following equivalent constraints
\begin{align}
& {\rm{C}}_5:\enspace \min_{  \{\Delta\bm{h}_{d,n,k}\}_{n=1}^N,\, \{\Delta\bm{Z}_{n,k}\}_{n=1}^N} \left \vert \sum_{n=1}^N \left (\bm{h}_{d,n,k}^H +\bm{v}^H \bm{Z}_{n,k} \right)\bm{w}_{n,k} \right \vert^2\geq \alpha_k^2, \enspace k=1,2,\cdots,K, \nonumber\\
& {\rm{C}}_6:\enspace \max_{  \{\Delta\bm{h}_{d,n,k}\}_{n=1}^N,\, \{\Delta\bm{Z}_{n,k}\}_{n=1}^N} \sum_{j\neq k} \left \vert \sum_{n=1}^N \left (\bm{h}_{d,n,k}^H +\bm{v}^H \bm{Z}_{n,k} \right)\bm{w}_{n,j} \right \vert^2+\sigma_k^2 \leq \beta_k, \enspace k=1,2,\cdots,K, \nonumber\\
& {\rm{C}}_7:\enspace\frac{\alpha_k^2}{\beta_k} \geq \gamma_k, \enspace k=1,2,\cdots,K,
\end{align}
where $\bm{\alpha}\triangleq [\alpha_1,\cdots,\alpha_K]^T$ and $\bm{\beta}\triangleq [\beta_1,\cdots,\beta_K]^T$ are slack variables to decompose the fractions. It is worth noting that the left-hand side (LHS) of $\mathrm{C}_7$ is a convex quadratic-over-linear function and hence $\mathrm{C}_7$ is a convex constraint.

Before handling the constraints $\mathrm{C}_5$ and $\mathrm{C}_6$, we first simplify them through the following definitions. To be specific, define the effective direct and cascaded channel related to the $k$th user as $\bm{h}_{d,k}\triangleq [\bm{h}_{d,1,k}^H,\cdots,\bm{h}_{d,N,k}^H]^H $ and $\bm{Z}_{k}\triangleq [\bm{Z}_{1,k},\cdots, \bm{Z}_{N,k}]$, respectively. Their estimates are denoted by $\hat{\bm{h}}_{d,k}$ and $\hat{\bm{Z}}_k$, respectively. According to (\ref{csimodel}) and \cite{green}, the uncertainties of $\bm{h}_{d,k}$ and $\bm{Z}_k$  follow
\begin{align}\label{eq15}
 \Vert \Delta \bm{h}_{d,k} \Vert_2 \leq \sqrt{\sum_{n=1}^N \epsilon_{d,n,k}^2} \triangleq \epsilon_{d,k}\enspace \mathrm{and}\enspace \Vert \Delta \bm{Z}_k \Vert_F \leq \sqrt{\sum_{n=1}^N \epsilon_{c,n,k}^2} \triangleq \epsilon_{c,k},
\end{align}
respectively. Now, the constraint in $\mathrm{C}_5$ is reformulated as
\begin{align}
\min_{ \Delta\bm{h}_{d,k},\,\Delta\bm{Z}_{k}}  \left \vert \left (\bm{h}_{d,k}^H +\bm{v}^H \bm{Z}_{k} \right)\bm{w}_{k} \right \vert^2\geq \alpha_k^2, \enspace k=1,2,\cdots,K,
\end{align}
where $\bm{w}_{k}\triangleq [\bm{w}_{1,k}^H,\cdots,\bm{w}_{N,k}^H]^H$. To further simplify the expression of $\mathrm{C}_6$, we define $\bm{W}_{-k}\triangleq [\bm{w}_1,\cdots,\bm{w}_{k-1},\bm{w}_{k+1},\cdots,\bm{w}_K]$ and the constraint in $\mathrm{C}_6$ is rewritten as
\begin{align}\label{eq17}
\max_{ \Delta\bm{h}_{d,k},\,\Delta\bm{Z}_{k}}  \left \Vert \left(\bm{h}_{d,k}^H +\bm{v}^H \bm{Z}_{k} \right)\bm{W}_{-k} \right \Vert_2^2\leq \beta_k-\sigma_k^2, \enspace k=1,2,\cdots,K.
\end{align}
Note that the constraints $\mathrm{C}_5$ and $\mathrm{C}_6$ are both infinite many constraints due to that CSI uncertainties $\Delta \bm{h}_{d,k}$ and $\Delta \bm{Z}_{k}$ are respectively lie in the regions $\Vert \Delta \bm{h}_{d,k}\Vert_2\leq \epsilon_{d,k}$ and $\Vert\Delta \bm{Z}_{k}\Vert_F \leq \epsilon_{c,k}$.
To deal with the infinitely many constraints, the S-procedure \cite{sp} is an effective technique to transform $\mathrm{C}_5$ and $\mathrm{C}_6$ into tractable forms of linear matrix inequalities (LMIs). For a conventional MIMO system with not-too-many antenna elements, the S-procedure based method achieves good performance with relatively low complexity. However, for a large-scale antenna system, especially for the considered large-size RISs, the scale of LMIs generated by the S-procedure become excessively large due to the introduction of a growing number of reflecting RIS elements, resulting in prohibitively high computational complexity. To this end, we alternatively introduce the following lemma before devising a low-complexity transformation. 

\begin{lemma}
For any choices of  $\bm{W}$ and $\bm{v}$, the worst-case value of LHS in constraint $\mathrm{C}_5$ is equivalent to
\begin{align}
&\min_{ \Delta\bm{h}_{d,k},\,\Delta\bm{Z}_{k}}  \left \vert \left (\bm{h}_{d,k}^H +\bm{v}^H \bm{Z}_{k} \right)\bm{w}_{k} \right \vert= \max \left \{\left \vert \left (\hat{\bm{h}}_{d,k}^H +\bm{v}^H \hat{\bm{Z}}_{k}\right )\bm{w}_{k} \right \vert -\left (\epsilon_{d,k} +\sqrt{ML}\epsilon_{c,k} \right )\Vert \bm{w}_{k} \Vert_2,0\right \},
\end{align}
under the CSI uncertainty regions, i.e., $\Vert \Delta \bm{h}_{d,k} \Vert_2\leq \epsilon_{d,k}$, and $\Vert \Delta \bm{Z}_{k} \Vert_F \leq  \epsilon_{c,k}$, for $n=1,2,\cdots,N$. Analogously, the maximal value of LHS in constraint $\mathrm{C}_6$ follows
\begin{align}
\max_{ \Delta\bm{h}_{d,k},\,\Delta\bm{Z}_{k}}  \left \Vert \left(\bm{h}_{d,k}^H +\bm{v}^H \bm{Z}_{k} \right)\bm{W}_{-k} \right \Vert_2 \leq \left \Vert (\hat{\bm{h}}_{d,k}^H +\bm{v}^H \hat{\bm{Z}}_{k})\bm{W}_{-k}\right \Vert_2+ \left(\epsilon_{d,k}+\sqrt{ML} \epsilon_{c,k} \right) \left \Vert \bm{W}_{-k} \right \Vert_F.
\end{align}
\end{lemma}
\begin{proof}
Please refer to Appendix A. \hfill $\blacksquare$
\end{proof}

We first equivalently rewrite constraint $\mathrm{C}_5$ by applying \emph{Lemma 1}. It follows
\begin{align}
\bar{\mathrm{C}}_5:\enspace \left \vert \left (\hat{\bm{h}}_{d,k}^H +\bm{v}^H \hat{\bm{Z}}_{k}\right )\bm{w}_{k} \right \vert - \epsilon_{k}\Vert \bm{w}_{k}\Vert_2 \geq \alpha_k,\enspace \alpha_k \geq 0,\enspace k=1,\cdots,K,
\end{align}
where $\epsilon_{k}\triangleq \epsilon_{d,k} +\sqrt{ML}\epsilon_{c,k}$ and we assume that $\left \vert  \left (\hat{\bm{h}}_{d,k}^H +\bm{v}^H \hat{\bm{Z}}_{k}\right )\bm{w}_{k} \right \vert \geq \epsilon_{k}\Vert \bm{w}_{k}\Vert_2 $. Although constraint $\bar{\mathrm{C}}_5$ is still nonconvex, we will show later that it is computationally efficient to deal with compared to solving large-scale LMIs. Similarly, according to \emph{Lemma 1}, we can also rewrite constraint $\mathrm{C}_6$ as
\begin{align}
\bar{\mathrm{C}}_6:\enspace \left \Vert (\hat{\bm{h}}_{d,k}^H +\bm{v}^H \hat{\bm{Z}}_{k})\bm{W}_{-k}\right \Vert_2+ \left(\epsilon_{d,k}+\sqrt{ML} \epsilon_{c,k} \right) \left \Vert \bm{W}_{-k} \right \Vert_F \leq \sqrt{\beta_k -\sigma_k^2}, \enspace k=1,\cdots,K.
\end{align}
It is worth noting that $\bar{\mathrm{C}}_6$ is a convex constraint.

Based on the above discussion, we recast the problem ($\mathrm{P1}$) as
\begin{align} \label{pp2}
(\mathrm{P}2):\enspace  \mathop{{\rm{maximize}}}_{\bm{W},\,\bm{v},\,\bm{\gamma}\,\bm{\alpha},\,\bm{\beta}} &\quad  \sum_{k=1}^K \log_2(1+\gamma_k) \nonumber \\
{\rm{s.t.}}&\quad {\rm{C}}_1,\enspace{\rm{C}}_2, \enspace \bar{\mathrm{C}}_3,\enspace \bar{\mathrm{C}}_5,\enspace \bar{\mathrm{C}}_6,\enspace \mathrm{C}_7.
\end{align}
Note that the problem ($\mathrm{P2}$) is still hard to solve due to some nonconvex constraints and the coupled variables, $\bm{W}$ and $\bm{v}$. But we are ready to introduce an effective solution based on AO in the next subsection.

\subsection{AO Algorithm for Worst-case Sum Rate Maximization in (P2)} 
For the coupled variables, we follow the popular AO framework \cite{zhaoya,zhaoya2} and optimize the precoding matrix, $\bm{W}$, and phase shifts, $\bm{v}$, in an alternating manner.

\subsubsection{Phase Shift Optimization}
Given the optimized precoding matrix, $\bm{W}$, we consider the optimization of phase shifts. The problem with respect to $\bm{v}$ reduces to
 \begin{align}\label{p3}
(\mathrm{P}3):\enspace\mathop{{\rm{maximize}}}_{\bm{v},\,\bm{\gamma},\,\bm{\alpha},\,\bm{\beta}} &\quad  \sum_{k=1}^K \log_2(1+\gamma_k) \nonumber \\
{\rm{s.t.}}&\quad {\rm{C}}_2, \enspace \bar{\mathrm{C}}_3,\enspace \bar{\mathrm{C}}_5,\enspace \bar{\mathrm{C}}_6,\enspace \mathrm{C}_7.
\end{align}

Firstly, the nonconvex modulus constraint $\mathrm{C}_2$ is a strict equality and is hard to deal with. To tackle this difficulty, we decompose constraint $\mathrm{C}_2$ into two equivalent constraints
\begin{align}
\mathrm{C}_{11}:\enspace \vert  [\bm{v}]_m \vert ^2 \leq 1,\enspace m=1,2,\cdots,ML,\nonumber \\
\mathrm{C}_{12}:\enspace  \vert [\bm{v}]_m \vert ^2 \geq 1,\enspace m=1,2,\cdots,ML.
\end{align}
Now, we find that constraint $\mathrm{C}_{11}$ is convex and $\mathrm{C}_{12}$ exhibits a nonconvex form. Moreover, constraints $\bar{\mathrm{C}}_3$ and $\bar{\mathrm{C}}_5$ are also nonconvex  difference of convex (DC) constraints. To handle these nonconvex constraints, we apply the P-CCP \cite{pccp} which is  proved to be effective in obtaining a feasible stationary point of the original problem \cite{ccp2}. Specifically, following the framework of P-CCP, we linearize the nonconvex parts of these constraints and reformulate them as
\begin{align}
&\quad  \tilde{\mathrm{C}}_3: \enspace\sum_{k=1}^K \frac{f_{\varpi}\left(\Vert \bm{w}_{n,k}\Vert_2^2\right) }{\left(1+\gamma_k^{(t)}\right)\mathrm{ln}2}\left( \gamma_k -\gamma_k^{(t)}\right)\leq\frac{C_n}{\xi_n}+[\bm{c}_1]_n,\enspace [\bm{c}_1]_n\geq 0, \enspace\forall n, \nonumber \\
&\quad \tilde{\mathrm{C}}_5: \enspace \frac{\Re\left\{\left( \hat{\bm{h}}_{d,k}^H + \bm{v}^{(t),H} \hat{\bm{Z}}_k \right)\bm{w}_k\bm{w}_k^H \left( \hat{\bm{h}}_{d,k}+  \hat{\bm{Z}}_k^H \bm{v} \right)\right\} }{\left\vert \left( \hat{\bm{h}}_{d,k}^H + \bm{v}^{(t),H} \hat{\bm{Z}}_k \right)\bm{w}_k\right \vert }-\epsilon_{k} \Vert \bm{w}_{k} \Vert_2\geq \alpha_k-[\bm{c}_2]_k,\enspace [\bm{c}_2]_k \geq 0,\,\forall k, \nonumber \\
&\quad \tilde{\mathrm{C}}_{11}:\enspace \vert  [\bm{v}]_m \vert ^2 \leq 1+[\bm{c}_3]_m,\enspace [\bm{c}_3]_m\geq 0, \enspace \forall m, \nonumber \\
&\quad \tilde{\mathrm{C}}_{12}:\enspace 2\Re \left \{ \left [\bm{v}^{(t)}\right ]_m^* [\bm{v}]_m \right \}-\left \vert \left [\bm{v}^{(t)}\right ]_m\right \vert^2\geq  1-[\bm{c}_3]_m,\enspace [\bm{c}_3]_m\geq 0, \enspace \forall m, 
\end{align}
where $\bm{c}_1$, $\bm{c}_2$, and $\bm{c}_3$ are newly introduced slack vectors with nonnegative elements.  Moreover, the penalty terms with respect to $\bm{c}_1$, $\bm{c}_2$, and $\bm{c}_3$ are imposed to the original objective function to penalize the violation of constraints, i.e., $\sum_{k=1}^K \log_2(1+\gamma_k)-\rho^{(t)}\left ( \Vert \bm{c}_1 \Vert_1+\Vert \bm{c}_2 \Vert_1 +\Vert \bm{c}_3 \Vert_1\right)$. In this way, after some basic manipulations, we obtain the transformed  problem as
\begin{align}\label{optv}
\mathop{{\rm{maximize}}}_{\bm{v},\,\bm{\gamma},\,\bm{\alpha},\,\bm{\beta}} &\quad \sum_{k=1}^K \log_2(1+\gamma_k)-\rho^{(t)}\left ( \Vert \bm{c}_1 \Vert_1+\Vert \bm{c}_2 \Vert_1 +\Vert \bm{c}_3 \Vert_1\right)\nonumber \\
{\rm{s.t.}}&\quad  \tilde{\mathrm{C}}_3,\enspace \tilde{\mathrm{C}}_5,\enspace \bar{\mathrm{C}}_6,\enspace \mathrm{C}_7,\enspace \tilde{\mathrm{C}}_{11},\enspace \tilde{\mathrm{C}}_{12},
\end{align}
which is a convex problem and thus can be efficiently solved by existing numerical convex program solvers, e.g., CVX tools \cite{cvxtool}. Moreover, $\rho^{(t)}$ is the penalty parameter in the $t$-th iteration for controlling the feasibility of the constraints. There exists an upper limit of $\rho^{(t)}$, i.e., $\rho^{(t)}\leq \rho^{\mathrm{max}}$, which is used to avoid numerical unstable problems.

In addition, it is worth noting that violation of the constraints is allowed due to the introduction of slack vectors in the initial stage of this algorithm. It is usually challenging to construct a feasible initial point due to the backhaul constraints. Therefore, we exploit the P-CCP not only for nonconvex constraints reformulation but also for effective initialization, which prompts this algorithm with efficient implementation.

We summarize the P-CCP for solving the problem $(\mathrm{P3})$ in Algorithm 1. It should be pointed out that for a sufficiently small $\varphi_1$, the violation is small enough such that the nonconvex constraint is satisfied when the iteration is terminated. Furthermore, the update of the penalty parameter follows the philosophy that a smaller $\rho^{(t)}$ is selected to pursue better performance at the initial stage and the nonconvex constraint is satisfied by increasing $\rho^{(t)}$ later. The algorithm is still guaranteed to converge to a stationary point \cite{ccp2} even though the objective function obtained in the iterative process is not strictly monotonic. 

\begin{algorithm}[!t]
\caption{P-CCP for Phase Shift Optimization}
\begin{algorithmic}[1]  
\STATE \textbf{Initialize} Iteration number $t=0$, feasibility tolerances $\varphi_1>0$ and $\varphi_2>0$, scale factor $\mu>1$, and penalty parameter $\rho^{(0)}$. Initialize $\bm{v}^{(0)}$, and $\bm{\gamma}^{(0)}$.
\STATE \textbf{Repeat}
\STATE $\quad$Solve the problem in (\ref{optv}). Output the solutions $\bm{v}^{(t+1)}$, $\bm{\gamma}^{(t+1)}$, and the slack vectors $\bm{c}_1$, $\bm{c}_2$, $\bm{c}_3$.
\STATE $\quad$Update $\rho^{(t+1)}=\min\left\{ \mu\rho^{(t)},\rho^{\mathrm{max}} \right \}$.
\STATE $\quad$Update  $t=t+1$.
\STATE \textbf{End if} $\Vert \bm{c}_1\Vert_1 +\Vert \bm{c}_2\Vert_1+\Vert \bm{c}_3\Vert_1\leq \varphi_1$ and $\left \Vert \bm{v}^{(t)}- \bm{v}^{(t-1)}\right \Vert_2 \leq \varphi_2$.
\end{algorithmic} 
\end{algorithm}

\subsubsection{Precoding Optimization}
While fixing the phase shifts, the optimization problem of the precoding matrix, $\bm{W}$, reduces to
\begin{align} \label{p2}
(\mathrm{P}4):\enspace \mathop{{\rm{maximize}}}_{\bm{W},\,\bm{\gamma},\,\bm{\alpha},\,\bm{\beta}} &\quad  \sum_{k=1}^K \log_2(1+\gamma_k) \nonumber \\
{\rm{s.t.}}&\quad {\rm{C}}_1, \enspace \bar{\mathrm{C}}_3,\enspace \bar{\mathrm{C}}_5,\enspace \bar{\mathrm{C}}_6,\enspace \mathrm{C}_7.
\end{align}

Firstly, the constraint in $\bar{\mathrm{C}}_3$ involves a bilinear function of $\bm{W}$ and $\bm{\gamma}$. By exploiting the fact that $4xy=(x+y)^2-(x-y)^2$ and introducing two sets of auxiliary variables $\bm{a}=\{a_{n,k}\vert \forall n,k\}$ and $\bm{b}=\{b_k\vert \forall k\}$ \cite{chen}, we equivalently rewrite $\bar{\mathrm{C}}_3$ as
\begin{align}
& {\rm{C}}_8:\enspace \sum_{k=1}^K \left[(a_{n,k}+b_k)^2-(a_{n,k}-b_k)^2\right]\leq\frac{4C_n}{\xi_n},\enspace a_{n,k},b_k\geq0,\enspace n=1,2,\cdots,N, \nonumber\\
& {\rm{C}}_9:\enspace f_{\varpi}\left (\Vert \bm{w}_{n,k}\Vert_2^2\right)\leq a_{n,k}, \enspace n=1,2,\cdots,N,\enspace k=1,2,\cdots,K, \nonumber\\
& {\rm{C}}_{10}:\enspace\log_2(1+\gamma_k)\leq b_k, \enspace k=1,2,\cdots,K.
\end{align}
Due to the existence of concave parts, constraints $\mathrm{C}_{8-10}$ are still nonconvex. To tackle the concave parts in constraints $\bar{\mathrm{C}}_5$, $\mathrm{C}_{8-10}$, we exploit the SCA technique to obtain a suboptimal solution. Specifically, through the Taylor expansion,  the concave parts can be replaced by their first-order upper bounds and then these nonconvex constraints become
\begin{align}\label{linear}
&{\hat{\mathrm{C}}}_5:\enspace \frac{\Re \left\{\bm{w}_k^{(t)}\hat{\bm{h}}_k \hat{\bm{h}}_k^H \bm{w}_k \right\}}{\left \vert \hat{\bm{h}}_k^H \bm{w}_k^{(t)} \right \vert }-\epsilon_{k} \Vert \bm{w}_{k} \Vert_2\geq \alpha_k, \enspace \alpha_k \geq 0,\enspace \forall k,\nonumber \\
&{\hat{\mathrm{C}}}_8:\enspace \sum_{k=1}^K \left[(a_{n,k}+b_k)^2 -2\left(a_{n,k}^{(t)}-b_k^{(t)} \right) (a_{n,k}-b_k)+\left(a_{n,k}^{(t)}-b_k^{(t)} \right)^2\right]\leq \frac{4C_n}{\xi_n},\enspace a_{n,k},b_k\geq0,\enspace \forall n, \nonumber\\
&{\hat{\mathrm{C}}}_9:\enspace f_{\varpi}\left(\Vert \bm{w}_{n,k}^{(t)}\Vert_2^2 \right) +\nabla  f_{\varpi}\left(\Vert \bm{w}_{n,k}^{(t)}\Vert_2^2 \right) \left(\Vert \bm{w}_{n,k}\Vert_2^2-\Vert \bm{w}_{n,k}^{(t)}\Vert_2^2\right)\leq a_{n,k},\enspace \forall n,k,\nonumber\\
&{\hat{\mathrm{C}}}_{10}:\enspace \frac{1}{\left(1+\gamma_k^{(t)}\right)\mathrm{ln}2}\left( \gamma_k -\gamma_k^{(t)}\right)\leq b_k,\enspace \forall k,
\end{align}
where $\hat{\bm{h}}_k \triangleq \hat{\bm{h}}_{d,k}+\hat{\bm{Z}}_{k}^H \bm{v}$ and $\bm{w}_k^{(t)}$, $\gamma_k^{(t)}$, $a_{n,k}^{(t)}$, and $b_k^{(t)}$are the optimal solutions obtained from the $t$-th iteration. \par

Now we are able to formulate the optimization problem in the $t$-th iteration for solving $(\mathrm{P}4)$ as the following convex problem, i.e.,
\begin{align}\label{optw}
\mathop{{\rm{maximize}}}_{\bm{W},\,\bm{\gamma},\,\bm{\alpha},\,\bm{\beta},\,\bm{a},\,\bm{b}} &\quad  \sum_{k=1}^K \log_2(1+\gamma_k) \nonumber \\
{\rm{s.t.}}&\quad {\rm{C}}_1, \enspace \hat{\mathrm{C}}_5,\enspace \bar{\mathrm{C}}_6,\enspace\mathrm{C}_7,\enspace \hat{\mathrm{C}}_8,\enspace \hat{\mathrm{C}}_9,\enspace \hat{\mathrm{C}}_{10}.
\end{align}
To summarize, we conclude the SCA algorithm for solving  $(\mathrm{P4})$ in Algorithm 2. Moreover, the feasible initial solutions of $\bm{W}^{(0)}$, $\bm{\gamma}^{(0)}$, $\bm{a}^{(0)}$, and $\bm{b}^{(0)}$ are constructed based on the optimization results of problem  $(\mathrm{P3})$. It is worth noting that the proposed SCA-based algorithm  is guaranteed to converge to a KKT solution of the nonconvex problem $(\mathrm{P4})$ \cite{sca}.
\begin{algorithm}[!t]
\caption{SCA for Precoding Optimization}
\begin{algorithmic}[1]  
\STATE \textbf{Initialize} Iteration number $t=0$ and feasibility tolerances $\varepsilon>0$. Initialize $\bm{W}^{(0)}$, $\bm{\gamma}^{(0)}$, $\bm{a}^{(0)}$, and $\bm{b}^{(0)}$ and calculate the initial objective value $\mathrm{Obj}^{(0)}$.
\STATE \textbf{Repeat}
\STATE $\quad$Solve problem in (\ref{optw}). Output the solutions $\bm{W}^{(t+1)}$, $\bm{\gamma}^{(t+1)}$, $\bm{a}^{(t+1)}$, and $\bm{b}^{(t+1)}$, and the objective value $ \mathrm{Obj}^{(t+1)}$.
\STATE $\quad$Update $t=t+1$.
\STATE \textbf{End if} $\left \vert {\rm{Obj}}^{(t)}-{\rm{Obj}}^{(t-1)} \right \vert\bigg/{\rm{Obj}}^{(t)} \leq \varepsilon $.
\end{algorithmic} 
\end{algorithm}

\subsubsection{Refinement}
It is worth mentioning that the approximation of $l_0$-norm in (\ref{appro}) is not very accurate in some special values and may lead to infeasible solutions. Specifically, if the $n$th AP does not serve the $k$th user, the obtained $\Vert \bm{w}_{n,k}\Vert_2^2$ via optimization would not be exactly zero and is replaced by a sufficiently small value, which may results in the violation of backhaul constraint at the $n$th AP. Hence, we perform the refinement similar to that in \cite{chen} to guarantee the feasibility of the solution. \par

In detail, we fix the user cluster for each AP according to the solutions obtained via AO iterations, and then perform another AO iteration to ensure a feasible solution. Accordingly, based on the obtained precoding matrix, we denote the user cluster at the $n$th AP by $\mathcal{S}_n=\left\{k\vert \Vert \bm{w}_{n,k}\Vert_2^2\geq \mu_{\mathrm{th}},\enspace \forall k \right\}$, where $\mu_{\mathrm{th}}$ is a predetermined small-value threshold. Accordingly, we reformulate constraint $\mathrm{C}_3$ as
\begin{align}
\mathrm{C}_{13}:\enspace \sum_{k\in \mathcal{S}_n } \log_2(1+\gamma_k) \leq \frac{C_n}{\xi_n},\enspace n=1,2,\cdots,N.
\end{align}
Hence, we reformulate the phase shift optimization subproblem $(\mathrm{P}3)$ as
\begin{align}\label{pr2}
(\mathrm{P}5):\enspace \mathop{{\rm{maximize}}}_{\bm{v},\,\bm{\gamma}\,\bm{\alpha},\,\bm{\beta}} &\quad  \sum_{k=1}^K \log_2(1+\gamma_k) \nonumber \\
{\rm{s.t.}}&\quad {\rm{C}}_2,\enspace \bar{\mathrm{C}}_5,\enspace \bar{\mathrm{C}}_6,\enspace \mathrm{C}_7, \enspace \mathrm{C}_{13},
\end{align}
which can also be handled with P-CCP. Similarly, the optimization subproblem of precoding matrix in $(\mathrm{P}4)$ is rewritten as 
\begin{align} \label{pr1}
(\mathrm{P}6):\enspace \mathop{{\rm{maximize}}}_{\bm{W},\,\bm{\gamma}\,\bm{\alpha},\,\bm{\beta}} &\quad  \sum_{k=1}^K \log_2(1+\gamma_k) \nonumber \\
{\rm{s.t.}}&\quad {\rm{C}}_1 ,\enspace \bar{\mathrm{C}}_5,\enspace \bar{\mathrm{C}}_6,\enspace \mathrm{C}_7, \enspace \mathrm{C}_{13}.
\end{align}
Note that the SCA technique is also useful for solving $(\mathrm{P}6)$.  By solving $(\mathrm{P}5)$ and $(\mathrm{P}6)$ alternatively, we arrive at a strictly feasible solution. The description of this algorithm to solve $(\mathrm{P}2)$ is elaborated in Algorithm 3.
\begin{algorithm}[!t]
\caption{AO-based Algorithm to Solve Problem $(\mathrm{P}2)$}
\begin{algorithmic}[1]  
\STATE \textbf{Initialize} Iteration number $i=0$, feasibility tolerances $\varepsilon$. Initialize $\bm{W}^{(0)}$ and $\bm{v}^{(0)}$.
\STATE \textbf{Repeat}
\STATE $\quad$Given $\bm{W}^{(i)}$, solve $(\mathrm{P}3)$ via P-CCP and output the solution  $\bm{v}^{(i+1)}$.
\STATE $\quad$Given $\bm{v}^{(i+1)}$, solve $(\mathrm{P}4)$ via SCA and output the solution $\bm{W}^{(i+1)}$ and the objective value $\mathrm{Obj}^{(i+1)}$.
\STATE $\quad$Update $i=i+1$.
\STATE \textbf{End if} $\left \vert {\rm{Obj}}^{(i-1)}-{\rm{Obj}}^{(i-2)} \right \vert\bigg/{\rm{Obj}}^{(i-1)} \leq \varepsilon $.
\STATE \textbf{Refinement}:
\STATE $\quad$Determine the user cluster for each AP based on $\bm{W}^{(i)}$.
\STATE$\quad$Solve  $(\mathrm{P}5)$ via P-CCP and  $(\mathrm{P}6)$ via SCA.
\end{algorithmic} 
\end{algorithm}

\subsection{Complexity Analysis}
According to \cite{pan}, since all the resulting convex problems involving second order cone (SOC) and affine constraints, a standard interior-point method (IPM) is effective to solve these problems. Based on the complexity analysis of IPM, the general computational complexity is given by
\begin{align}\label{complex}
\mathcal{O}\left (\sqrt{2I} \left(n^3+n^2 \sum_{i=1}^I a_i^2\right )\right ),
\end{align}
where $n$ is the number of variables, and $I$ is the number of SOC of size $a_i$. We consider a typical setup where the total number of reflecting elements is much larger than the number of users, i.e., $ML\gg K$. Then in Table \ref{table_complex}, we summarize the complexity analysis of solving these problems.  From this table, it concludes that the computational complexity of the proposed robust design is at the order of 
\begin{align}\label{c1}
\mathcal{O} \left(K^{5.5}N^{4.5}N_t^4+ KM^{4.5}L^{4.5} \right).
\end{align}
On the other hand, in the S-procedure based method \cite{pan,hu}, the constraints $\bar{\mathrm{C}}_5$ and $\bar{\mathrm{C}}_6$ are LMIs with size $(ML+1)NN_t+1$ and $K+2NN_t$, respectively. Hence, the computational complexity of the S-procedure based method is at the order of 
\begin{align}\label{c2}
\mathcal{O} \left(K^{2.5}N^{4.5}N_t^4 M^3 L^3+ KN^3 N_t^3 M^{4.5}L^{4.5} \right).
\end{align}
By comparing (\ref{c1}) and (\ref{c2}), it is obvious that the proposed approximation method has a much lower complexity.

\begin{table*}[!t]
\caption{Detailed complexity analysis of the proposed algorithm}
\label{table_complex}
    \centering
	\begin{tabular}{|p{1.8cm}<{\centering}|p{2.7cm}<{\centering}|p{1.8cm}<{\centering}|p{1.8cm}<{\centering}|p{1.8cm}<{\centering}|p{1.8cm}<{\centering}|p{1.8cm}<{\centering}|}	
        \hline
		\multirow{5}*{\makecell{Phase shift\\Optimization\\$(\mathrm{P}3)$, $(\mathrm{P}5)$}} & SOC constraints &  $\tilde{\mathrm{C}}_5$ &  $\bar{\mathrm{C}}_6$ & $\tilde{\mathrm{C}}_{11}$ & $\tilde{\mathrm{C}}_{12}$ & / \tabularnewline
		\cline{2-7}
		~&Number of SOCs &  $K$ & $K$ & $ML$ &  $ML$&/\tabularnewline
		\cline{2-7} 
		~&Size &  $ML$&  $ML$ &  $1$ &  $1$ &/\tabularnewline
		\cline{2-7}
		~&Number of variables &  \multicolumn{5}{c|}{$ML$}\tabularnewline
		\cline{2-7}
		~&Overall complexity &  \multicolumn{5}{c|}{$\mathcal{O} \left( KM^{4.5}L^{4.5} \right)$}\tabularnewline
\hline
		\multirow{5}*{\makecell{Precoding\\Optimization\\$(\mathrm{P}4)$, $(\mathrm{P}6)$}} & SOC constraints & $\mathrm{C}_1$ & $\hat{\mathrm{C}}_5$&  $\bar{\mathrm{C}}_6$ &  $\hat{\mathrm{C}}_8$
&  $\hat{\mathrm{C}}_9$ \\
		\cline{2-7}
		~&Number of SOCs &$N$ &  $K$ &  $K$ & $N$ & $NK$\tabularnewline
		\cline{2-7} 
		~&Size &  $KN_t$ & $NN_t$&  $(K-1)NN_t$ & $K$ & $N_t$\tabularnewline
		\cline{2-7}
		~&Number of variables &  \multicolumn{5}{c|}{$KNN_t+KN+K$}\tabularnewline
		\cline{2-7}
		~&Overall complexity &  \multicolumn{5}{c|}{$\mathcal{O} \left(K^{5.5}N^{4.5}N_t^4 \right)$}\tabularnewline
        \hline
	\end{tabular}%
\end{table*}

\section{Simulation Results}
\subsection{Simulation Setup}
In this section, we provide simulation results to demonstrate the effectiveness of the proposed algorithms. We consider a system with the topology as depicted in Fig. 2. Four APs located at $(100,0), (-100,0), (0,100)$, and $(0,-100)$, respectively, are deployed in this network. Moreover, four RISs are randomly dropped with a uniform distribution in this area of interest, i.e., following the Poisson point process. The $K$ users are randomly dropped in a circle with radius $R$. The heights of APs, RISs, and users are set to 20\,m, 5\,m, and 1.5\,m, respectively.  The channels between each AP and the users are assumed as Rayleigh channels, while the cascaded channels are modelled as Rician channels, which are given by
\begin{equation}
\bm{H}=\sqrt{\beta} \left (\sqrt{\frac{\kappa}{\kappa+1}}\bm{H}^{{\rm{LoS}}}+ \sqrt{\frac{1}{\kappa+1}}\bm{H}^{{\rm{NLoS}}}\right ),
\end{equation}
where $\beta$ denotes the large-scale path loss, $\kappa \geq 0$ is the Rician factor, $\bm{H}^{{\rm{LoS}}}$ and $\bm{H}^{{\rm{NLoS}}}$ represent the line-of-sight (LoS) component and non-LoS (NLoS) component, respectively. The columns of $\bm{H}^{{\rm{NLoS}}}$ follow the complex Gaussian distribution with a zero mean and unit variance. Note that Rayleigh channels contain only NLoS components. In addition, when the distance between RIS and user exceeds a certain value, we assume that due to the presence of obstacles, there is no LoS path. The large-scale path loss $\beta$ is formulated as $\beta=\beta_0 \left(d/d_0 \right)^{-\alpha}$, where $d_0=1\,$m is the reference distance, $\beta_0=-30\,$dB is the path loss at the reference distance, and $\alpha$ is the path-loss exponent\cite{pan3}. The path-loss exponents of the Rician channels and Rayleigh channels are denoted by $\alpha_{1}$ and $\alpha_{2}$, respectively. For ease of presentation, we define $\delta_{d}$, $\delta_{c}$  as the uncertainty level of direct and cascaded channels, respectively, which follow $\delta_d=\Vert\Delta \bm{h}_{d,k}\Vert_2 /\Vert \bm{h}_{d,k} \Vert_2$ and $\delta_c=\Vert\Delta \bm{Z}_{k}\Vert_F /\Vert \bm{Z}_{k} \Vert_F$, $\forall k$. For simplicity, we set the maximum power of all the APs to be $P$, and the maximum capacity of all backhauls to be $C$. Unless otherwise specified, the other parameters are listed in Table \ref{table:1}.

\begin{table}[!t]   
\begin{center}   
\caption{Simulation Parameters}  
\label{table:1} 
\begin{tabular}{|l|c|l|c|}   
\hline    Number of transmit antennas of each AP, $N_t$ &4& Parameter of backhaul capacity margin $\xi_n,\enspace \forall n$ & 1.1 \\ 
\hline    Number of reflecting elements of each RIS, $M$ &16 &Parameter of approximation accuracy, $\varpi$  &$10^{-3}$ \\  
\hline    Number of users, $K$ &4 &SCA and AO convergence tolerance, $\varepsilon$ & $10^{-4}$ \\
\hline    Cell radius, $R$ & $100\,$m&Thresholds for P-CCP convergence, $\varphi_1$ and $\varphi_2$ &$10^{-3}$ \\
\hline    Available bandwidth & $10\,$MHz&Threshold for refinement, $\mu_{\mathrm{th}}$ &$10^{-3}$ \\
\hline    Noise power, $\sigma_k^2,\enspace \forall k$ & $-80\,$dBm &Rician factor, $\kappa$ &3 \\
\hline    Maximum transmit power, $P$& $30\,$dBm& Path-loss exponent of direct channels, $\alpha_1$ &3.75\\
\hline    Maximum backhaul capacity, $C$& $200\,$Mbps& Path-loss exponent of cascaded channels, $\alpha_2$ &2.2\\
\hline
\end{tabular}   
\end{center}   
\end{table}

\begin{figure}[!t]
\centering
\includegraphics[width=3.6in]{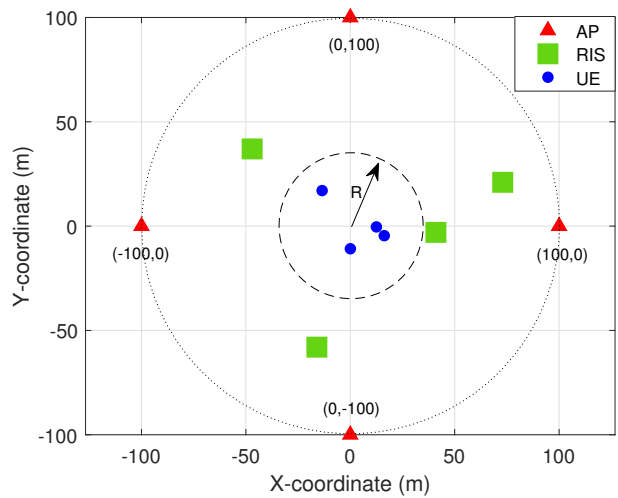}
\caption{ The simulated RIS-aided CF system.}
\end{figure}

We mainly compare our proposed algorithms with the following five baseline schemes.
\begin{itemize}
\item Non-robust RIS-CF: the imperfect CSI is treated as perfect for beamforming design to highlight the potential performance degradation caused by CSI imperfectness.
\item RandPhase: we only optimize the precoding vectors and randomly select the phase shifts at RISs to evaluate the importance of phase shift optimization.
\item CF w/o RIS \cite{cf1}: we consider a traditional CF system without the deployment of RISs.
\item SC-CF \cite{cf2}: we consider the conventional small-cell scheme with the assistance of RISs, in which each AP serves only serves a specific user.
\item Centralized BS \cite{zhaoya}: we consider a centralized base station (BS) to serve all the users in the area of interest without any backhaul constraints.
\end{itemize}

\subsection{Convergence of the Proposed Algorithm}
We first consider verifying the convergence of the proposed Algorithm 3 in Fig. 3. Here, we set that $\delta_d=\delta_c=\delta$, $N=4$, $N_t=2$, $M=10$, and $L=2$. As a comparison, we consider a baseline method where constraints $\bar{\mathrm{C}}_5$ and $\bar{\mathrm{C}}_6$ are replaced by large-scale LMIs generated by the S-procedure \cite{pan,hu}. Firstly, it is observed that the proposed algorithm is much more computational efficient and it is order-of-magnitude faster than the S-procedure-based method. On the other hand, compared with the S-procedure-based method, the worst-case sum rate obtained by the proposed algorithm is almost the same. For instance, a marginal performance loss of  6\% is observed  when $\delta=0.1$. It is verified that the proposed method achieves comparable performance with the S-procedure-based method even for relatively large uncertainty regions. Especially for stringent backhaul constraint of $C=100$ Mbps, both methods exhibit almost the same performance. These results further validates that the proposed algorithm achieves excellent performance with much lower complexity and faster convergence especially in the the RIS-aided CF system with strict backhaul constraints.

\begin{figure}[!t]
\centering
\includegraphics[width=3.3in]{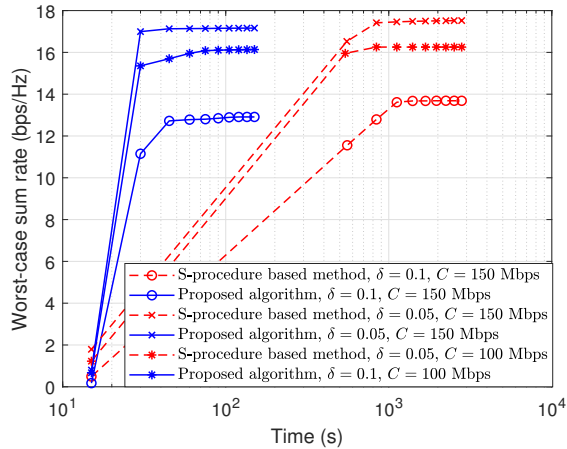}
\caption{ Convergence of the proposed algorithm.}
\end{figure}

\subsection{Impact of the Maximum Transmit Power of AP}
Next, we focus on the impact of the maximum transmit power of each AP in Fig. 4. The channel uncertainties are setted as $\delta_d=0.02$ and $\delta_c=0.04$. As shown in Fig. 4, with the increase of the transmit power of each AP, the performance of each scheme enhances significantly and the proposed algorithm consistently outperforms all the baseline schemes. Compared with non-robust design, we find that the performance gap is enlarged with an increasing power budget, which highlights the importance of robust design against CSI errors. By contrast, it is seen that performance gap between the proposed algorithm and RandPhase algorithm decreases with the increases of the transmit power. This is because the backhaul capacity is close to the upper bound when the transmit power is sufficient large and phase shift optimization can only offer marginal gain. Moreover, compared with the traditional CF system without RISs, a large number of independent controllable paths are created through the deployment of RISs in the proposed scheme and thus bring significant diversity gain. Furthermore, the cooperative APs are much more robust to dynamic wireless environment than the uncooperative ones in small-cell systems, which results in better performance. However, under the high transmit power budget, the performance gap is reduced. Indeed, each AP in small-cell systems only serves a specific user and the backhaul constraints are more relaxed compared with the RIS-aided CF systems, which result in the greater performance growth.

\begin{figure}[!t]
\centering
\includegraphics[width=3.3in]{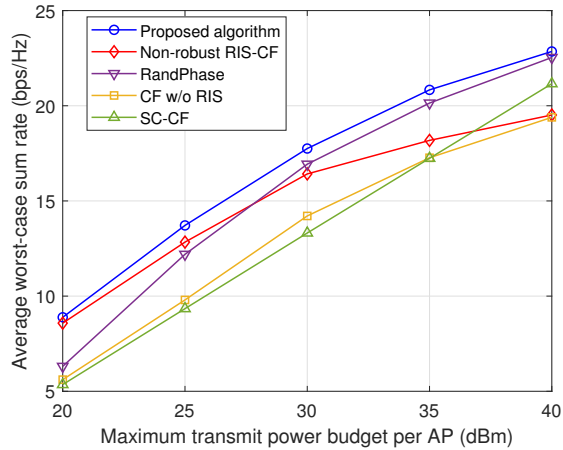}
\caption{ Average worst-case sum rate versus the maximum transmit power of each AP.}
\end{figure}

\subsection{Impact of Backhaul Capacity}

\begin{figure}[!t]
\centering
\includegraphics[width=3.3in]{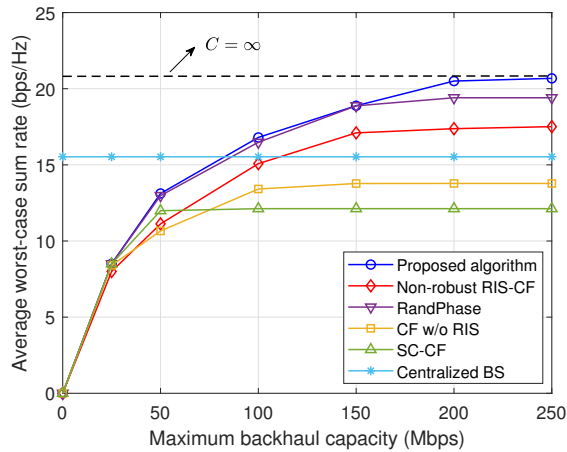}
\caption{ Average worst-case sum rate versus the maximum backhaul capacity.}
\end{figure}

In Fig. 5, we  depict the average worst-case sum rate versus the maximum backhaul capacity with $\delta_d=0.02$ and $\delta_c=0.04$. For all the considered schemes, with the increase of backhual capacity, the average worst-case sum rate first increases rapidly and then gradually tends to an upper bound. With a small backhaul capacity, only a small amount data can be conveyed to all the APs for effective beamforming or each AP can only serve a subset of users. In contrast, under a sufficient capacity budget, all the APs are able to participate in the service for each user and the worst-case sum rate is mainly limited by the transmit power and the CSI errors. Note that due to the limited backhaul, there is little performance difference between the proposed algorithm and RandPhase algorithm until the backhaul capacity is sufficient. Indeed, with more stringent backhaul constraints, each AP can only serve a specific user, in which case there does not exist any difference between the proposed CF system and small-cell system. Hence, this explains why the proposed scheme has similar performance to small-cell system when the maximum backhaul capacity is less than $50\,$Mbps. In addition, when the backhaul capacity $C$ is less than 80 Mbps, the centralized BS exhibits a superior performance as an upper bound. On the other hand, the CF architecture is preferred with adequate backhaul capacity.

\subsection{Impact of CSI Uncertainty}
\begin{figure}[!t]
\subfigure[]{
\begin{minipage}[t]{1\linewidth}
  \centering
  \includegraphics[width=3.3in]{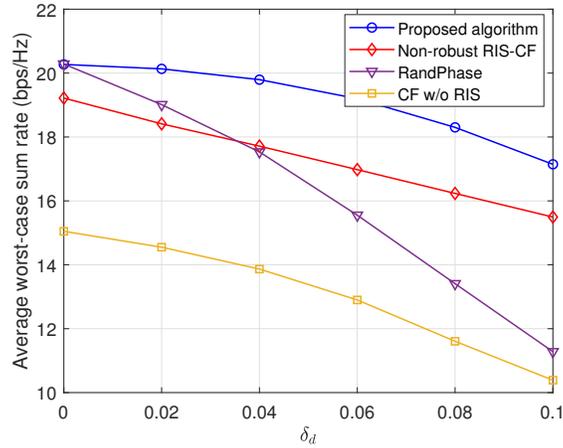}
 \end{minipage}
}

\subfigure[]{
\begin{minipage}[t]{1\linewidth}
  \centering
  \includegraphics[width=3.3in]{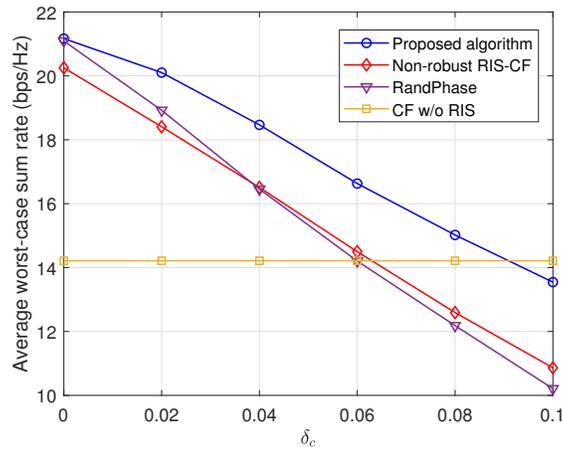}
 \end{minipage}
}
\caption{ (a) Average worst-case sum rate versus the uncertainty level of direct channels; (b) Average worst-case sum rate versus the uncertainty level of cascaded channels.}
\end{figure}
Fig. 6 depicts the average worst-case sum rate versus the uncertainty level of both the direct and cascaded channels. We set the maximum backhaul capacity is $200\,$Mbps. It is observed that when the uncertainty level of the direct channel $\delta_d$ increases from 0 to 0.1, the worst-case sum rate of the proposed algorithm decreases by only about 20 percent. By contrast, with the decrease of uncertainty level of the cascade channel $\delta_c$, the worst-case sum rate decreases sharply. This is because the number of reflecting elements at RISs is much larger than the number of antennas at the APs. As a result, the CSI error of the cascaded channel dominates the potential diversity gain and causes a more severe performance loss. The result illustrates that the uncertainty level of the cascaded channel has a greater impact on the worst-case sum rate than that of the direct channel. In other words, even if we have accurate estimation of the direct channel, a large uncertainty level of the cascaded channel can lead to a sharp decline in the sum rate. 

Moreover, we find that although RandPhase algorithm has similar performance as the proposed algorithm when the uncertainty level is low, its performance deteriorates severely as the uncertainty level of CSI error increases. Indeed, the performance of RandPhase algorithm is sensitive to the CSI errors of both the direct channel and the cascaded channel, as the degraded CSI estimation quality magnifies the mismatches in resource allocation. This illustrates the necessity of RIS phase shift optimization under large CSI errors. Then, compared with the traditional CF system, it can be seen that when $\delta_c$ increases to 0.09, the deployment of RISs cannot bring any performance gain and even plays a negative role. This further shows that having accurate CSI of cascaded channels is the key to unlock the potential of RIS.

\subsection{Number of Reflecting Elements}
\begin{figure}[!t]
\centering
\includegraphics[width=3.3in]{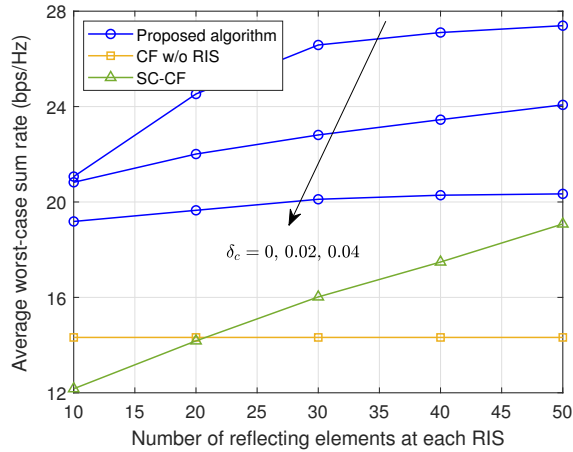}
\caption{  Average worst-case sum rate versus the number of reflecting elements.}
\end{figure}

In Fig. 7, the relationship between the worst-case sum rate and the number of reflection elements at each RIS is discussed. The uncertainty level of the direct channel $\delta_d$ is fixed as 0.02 in the following simulations because its impact on worst-case sum rate is relatively small. For the schemes with the deployment of the RISs, the performance increases monotonically as the number of reflecting elements increases. However, for larger CSI errors, i.e., $\delta_c=0.04$, the performance gain in having more reflecting elements is limited. In fact, the diversity gain brought by increasing the number of RIS reflecting elements is neutralized by the performance degradation brought by the larger CSI errors. Therefore, it is not necessarily economical to employ a large number of reflecting elements in the case of large CSI errors. Moreover, despite the availability of perfect CSI, i.e., $\delta_c=0$, when sufficient reflecting elements have been deployed, the performance gain introduced by increasing the number of reflecting elements is still limited. This is because the limited backhaul capacity hinders the growth of the worst-case sum rate. Hence, even if we have accurate CSI, increasing the number of reflecting elements may not introduce a significant improvement in performance under a low backhaul capacity.\par

\subsection{Number of RISs}
\begin{figure}[!t]
\centering
\includegraphics[width=3.3in]{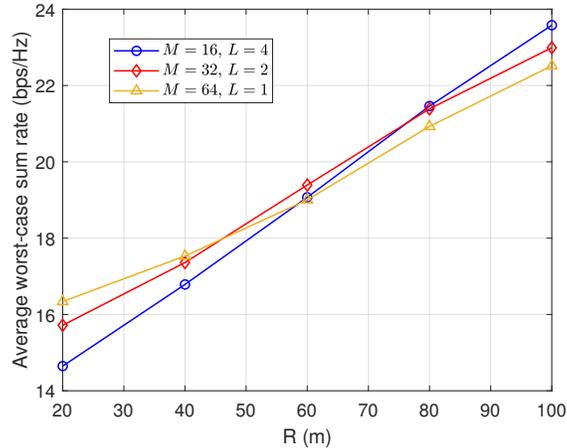}
\caption{ Average worst-case sum rate versus the number of RISs.}
\end{figure}

In Fig. 8, we investigate the influence brought by the number of RISs when the total number of reflecting elements is fixed. We can find that when the users are relatively close to each other, i.e., when the radius $R$ is small, the centrally deployed RIS demonstrates obvious performance advantages over distributed RISs. By contrast, when the user distribution becomes more and more dispersed over the service area, adopting distributed RISs can offer a higher worst-case sum rate. In particular, distributed RISs are likely to create more strong end-to-end LoS paths when the users are scattered in a larger area \cite{gao}.

\section{Conclusion}
In this paper, we investigated the robust beamforming design for a RIS-aided CF system with the consideration of CSI uncertainties at the transmitter and the capacity-limited backhaul. The precoding at the APs and the phase shifts at the RISs were jointly optimized for maximizing the worst-case sum rate. To address the nonsmooth constraints and semi-infinite constraints, we proposed a computational-efficient transformation scheme to pave the way for the development of an iterative suboptimal algorithm based on AO. The the P-CCP and the SCA method were exploited for  RIS phase shift and precoding optimization, respectively. Numerical results confirmed excellent performance of the proposed algorithm in the presence of channel errors and further show the importance of the cascaded channel estimation and RIS phase shift optimization, together with the advantages of decentralized deployment of RISs in CF systems. Additionally, beamforming designs considering a more realistic cascaded channel model \cite{george5} and relying on less CSI \cite{george6} should be of our interest in the future work.

\appendices
\section{Proof of Lemma 1}
By applying the triangle inequality and the Cauchy-Schwarz inequality to the LHS of constraint $\mathrm{C}_5$, we obtain
\begin{align}\label{obj}
&\left \vert  \left ((\hat{\bm{h}}_{d,k}+\Delta \bm{h}_{d,k} )^H +\bm{v}^H (\hat{\bm{Z}}_{k}+ \Delta \bm{Z}_{k})\right )\bm{w}_{k} \right \vert \nonumber \\
&\overset{\text{(a)}}{\geq} \left \vert  \left (\hat{\bm{h}}_{d,k}^H +\bm{v}^H \hat{\bm{Z}}_{k}\right )\bm{w}_{k}\right \vert -\left \vert \left (\Delta \bm{h}_{d,k}^H +\bm{v}^H \Delta \bm{Z}_{k}\right )\bm{w}_{k} \right \vert \nonumber\\
&\overset{\text{(b)}}{\geq} \left \vert  \left (\hat{\bm{h}}_{d,k}^H +\bm{v}^H \hat{\bm{Z}}_{k}\right )\bm{w}_{k} \right \vert -\left (\epsilon_{d,k}+ \left \Vert \bm{v}^H \Delta \bm{Z}_{k}  \right \Vert_2 \right )\Vert \bm{w}_{k} \Vert_2\nonumber \\
&\overset{\text{(c)}}{\geq} \left \vert  \left (\hat{\bm{h}}_{d,k}^H +\bm{v}^H \hat{\bm{Z}}_{k}\right )\bm{w}_{k} \right \vert -\left (\epsilon_{d,k}+ \sqrt{ML}\epsilon_{c,k}  \right )\Vert \bm{w}_{k} \Vert_2
\end{align}
where the inequality in (a) is due to the triangle inequality,  the inequality in (b) comes from the Cauchy-Schwarz inequality and $\Vert \Delta \bm{h}_{d,k} \Vert_2\leq \epsilon_{d,k}$, and the inequality in (c) exploits the fact that $\Vert \Delta \bm{Z}_{k} \Vert_2\leq \epsilon_{c,k}$ and $\Vert \bm{v} \Vert_2=\sqrt{ML}$. If $ \left \vert  \left (\hat{\bm{h}}_{d,k}^H +\bm{v}^H \hat{\bm{Z}}_{k}\right )\bm{w}_{k} \right \vert \geq \left (\epsilon_{d,k}+ \sqrt{ML}\epsilon_{c,k}  \right )\Vert \bm{w}_{k} \Vert_2$, it is checked that the lower bound in (\ref{obj}) is achieved when 
\begin{align}
&\Delta \bm{h}_{d,k}=-\epsilon_{d,k} e^{j\theta} \frac{\bm{w}_{k}}{\Vert \bm{w}_{k}\Vert_2},\nonumber \\
&\Delta \bm{Z}_{k} =-\epsilon_{c,k} e^{j\theta}  \frac{\bm{v}\bm{w}_{k}^H}{\sqrt{ML}\Vert \bm{w}_{k}\Vert_2},
\end{align}
where $\theta \triangleq \angle \left (   \left (\hat{\bm{h}}_{d,k}^H +\bm{v}^H \hat{\bm{Z}}_{k}\right )\bm{w}_{k}  \right )$.
When $\left \vert \left (\hat{\bm{h}}_{d,k}^H +\bm{v}^H \hat{\bm{Z}}_{k}\right )\bm{w}_{k} \right \vert <  \left (\epsilon_{d,k}+ \sqrt{ML}\epsilon_{c,k} \right )\Vert \bm{w}_{k} \Vert_2$, we select
\begin{align}
\Delta \bm{h}_{d,k}&=-\frac{\epsilon_{d,k}}{\epsilon_{d,k}+\sqrt{ML} \epsilon_{c,k}} \frac{\bm{w}_{k}\bm{w}_{k}^H \hat{\bm{h}}_{k}}{\Vert \bm{w}_{k}\Vert_2^2 },\enspace n=1,\cdots,N,\nonumber\\
\Delta \bm{Z}_{k}&=-\frac{\epsilon_{c,k}}{\sqrt{ML}\epsilon_{d,k}+ML \epsilon_{c,k}} \frac{\bm{v} \hat{\bm{h}}_{k}^H \bm{w}_{k}\bm{w}_{k}^H }{\Vert \bm{w}_{k}\Vert_2^2 },\enspace n=1,\cdots,N.
\end{align}
It is easily verified that $\Vert \Delta \bm{h}_{d,k} \Vert_2 \leq \epsilon_{d,k}$, $\Vert \Delta \bm{Z}_{k} \Vert_F \leq \epsilon_{c,k}$, and $\left \vert \left (\bm{h}_{d,k}+\bm{v}^H \bm{Z}\right )\bm{w}_{k} \right \vert=0$. Then, we can conclude that the worst-case value of $\left \vert \left (\bm{h}_{d,k}+\bm{v}^H \bm{Z}\right )\bm{w}_{k} \right \vert$ is equal to $\max\left\{ \left \vert \left (\hat{\bm{h}}_{d,k}^H +\bm{v}^H \hat{\bm{Z}}_{k}\right )\bm{w}_{k} \right \vert\right.$ $\left. - \left (\epsilon_{d,k}+ \sqrt{ML}\epsilon_{c,k} \right )\Vert \bm{w}_{k} \Vert_2, 0\right\}$.

As for the maximal value of the LHS of constraint $\mathrm{C}_6$, we also have
\begin{align}\label{obj2}
&\left \Vert  \left ((\hat{\bm{h}}_{d,k}+\Delta \bm{h}_{d,k} )^H +\bm{v}^H (\hat{\bm{Z}}_{k}+ \Delta \bm{Z}_{k})\right )\bm{W}_{-k} \right \Vert_2 \nonumber \\
&\leq \left \Vert  \left (\hat{\bm{h}}_{d,k}^H +\bm{v}^H \hat{\bm{Z}}_{k}\right )\bm{W}_{-k}\right \Vert_2 +\left \Vert \left (\Delta \bm{h}_{d,k}^H +\bm{v}^H \Delta \bm{Z}_{k}\right )\bm{W}_{-k} \right \Vert_2 \nonumber\\
&\overset{\text{(a)}}{\leq} \left \Vert  \left (\hat{\bm{h}}_{d,k}^H +\bm{v}^H \hat{\bm{Z}}_{k}\right )\bm{W}_{-k}\right \Vert_2 +\left (\epsilon_{d,k}+ \sqrt{ML}\epsilon_{c,k}  \right )\Vert \bm{W}_{-k} \Vert_F,
\end{align}
where the inequality in $(a)$ is similar to that derived in (\ref{obj}). Hence, we complete the proof.

\end{document}